\documentclass[preprint,12pt]{elsarticle}

\usepackage{amssymb}

\usepackage{amsthm}
\usepackage{amsmath}
\usepackage{amssymb}
\usepackage{mathrsfs}
\usepackage{color}
\usepackage{dsfont}
\usepackage{graphics}
\usepackage{graphicx}

\newtheorem{monTh}{Theorem}

\begin{document}

\begin{frontmatter}



\title{A New Distribution Version of Boneh-Goh-Nissim  Cryptosystem:\\
Security and performance analysis}

 \author[label1]{Oualid Benamara }
 \address[label1]{University of Science and Technology Houari Boumediene,
Institute of mathematics}
\ead{benamara.oualid@gmail.com}
\author[label2]{Fatiha Merazka }
 \address[label2]{University of Science and Technology Houari Boumediene,
Institute of electronics}
 \ead{fmerazka@usthb.dz}

\begin{abstract}
The aim of this paper is to provide two distributed versions
of Boneh elliptic curve cryptography (BECC) algorithm. We give a proof of semantic security for the first one. This  guaranties that our algorithm is semantically secure in the contest of active non-adaptive adversary. Furthermore, we prove that the second version of our distributed scheme is computationally
more efficient in time computation than ElGamal elliptic curve threshold cryptosystem and secure under
the Subgroup Decision Problem assumption.

\end{abstract}



\begin{keyword}
Elliptic curve  cryptography\sep Threshold scheme\sep ElGamal cryptosystem.



\end{keyword}

\end{frontmatter}
\section{Introduction}

Public key cryptosystems are widely used nowadays in electronic banking, online browsing,
voting systems and so on. While symmetric systems are more efficient than the asymmetric ones, those later are suitable 
for key handling in an intrusted setup.
A
well known drawback in public key cryptography, is that the knowledge of the secret
keys (like a server) provide full control, and may be viewed as a weakness. Hence a new direction in research
named distributed systems. Instead of relying on a single trusted party, the key is distributed among $n$ parties and at least $m$ parties must collaborate in order to recover the secrete key. $m$
is chosen as a parameter of the system, at the setup level.

The definition of semantic security of threshold cryptosystem used in this paper follows the one
stated in \cite{Fouque}.

The decryption operation of the BGNC (?) 
used in this paper is done by computing the discrete logarithm on elliptic
curve. We use \emph{Pollard Lambda} algorithm (reference) which has
complexity $O(\sqrt{T})$ in time of computation, where $T$
is the length of the interval to which the message $m$ belongs.

Our second algorithm Boneh elliptic curve cryptography (BECC) is more efficient than EECC in terms of computation efficiency.
We evaluate the computational time of our algorithm and compare it to EECC\cite{Farras}. 

In our paper, we focus on encryption and decryption time.

\subsection{Organisation}
This paper is organized as follows. In Section \ref{def}, we recall the basic definition and
the scenario when dealing with threshold cryptography. We present our novel threshold scheme together
with a proof of security of this later in Section \ref{ecc}. The other distribution version together with its complexity analysis are in Section \ref{eccb}.

The security of the Paillier, ElGamal ECC (EECC) and BGNC cryptosystems rely on the  \emph{Decisional Composite Residuosity} assumption (DCRA) \cite{Paillier}, the \emph{Elliptic Curve Discrete Logarithm Problem} (ECDLP) \cite{Bena}  and the \emph{ Subgroub Decision
Problem} assumptions \cite{boneh} respectively. So, building on the ECDLP and the SDP assumptions, we prove that
our distributed scheme is semantically secure, in the meaning introduced in \cite{Fouque}. For the best of our knowledge, this has never been done. 
\section{Related work}
Several elliptic curve distributed systems have been proposed in
the literature \cite{Fouque}, \cite{Fournaris}, \cite{Ertaul},
\cite{thresholdecc}, \cite{zhang}. Our paper aims to introduce two
novel distributed schemes for the BGNC
 introduced in \cite{boneh} following the paradigm of
\cite{Fouque}. Fouque et al. \cite{Fouque} proposed a first
distribution of the Paillier system together with a proof of semantic
security. In \cite{Fournaris}, Fournaris proposed a distributed
version of ECC. However, no proof of security was given.
\section{Preliminaries}
\subsection{Zero Knowledge (ZK) Proof of Equality}
\label{zkp}
In this section we will recall the ZK proof of equality algorithm given in \cite{zhang} to produce proofs of validity in step 3
of our threshold ECC, see Section \ref{ecc}.

Suppose that party B needs to prove
that it has a valid secret share key $s$ to party A without divulging
any information about $s$.
\begin{enumerate}
\item A selects two random integers $a$ and $b$ smaller than
the order of the cyclic group generated by a point of an elliptic
curve.
\item It computes $(a_{1},a_{2})=asg \mod p$ and $(b_{1},b_{2})=bg \mod p$,
where $p$ is a prime number related to the finite field $GF(p)$
on which the elliptic curve $E$ is constructed and $g$ is a generator of
the group $G$ of the points on $E$.
\item It also computes $t_{1}=a_{1}b_{1} \mod p$ and  $t_{2}=a_{2}b_{2} \mod p$
and sends $(ag, t_{1}, t_{2})$ to party B.
\item B computes $sag=(r_{1},r_{2})\mod p $.
\item B computes $z_{1}=t_{1}r_{1}^{-1}\mod p$ and $z_{2}=t_{2}r_{2}^{-1} \mod p$
and sends $(z_{1},z_{2})$ to A.

\end{enumerate}

At the end, A verifies that $(b_{1},b_{2})=(z_{1},z_{2})$. If this
is true, B has the secret $s$. Otherwise, the proof $proof=
(z_{1},z_{2})$ is not valid.

The correctness of this scheme can be verified from the Menezes--Vanstone
elliptic curve cryptosystem \cite{zhang}.
section{Threshold Cryptosystems}
\label{def}
\subsection{Formal Definition}
A threshold cryptosystem consists of the following four components:
\begin{itemize}
\item A key generation algorithm takes as input a security parameter $k$, the number
$l$ of decryption servers, the threshold parameter $t$ and a random
string $\omega$. It outputs a public key $PK$, a list $PK_{1},\ldots
,PK_{l}$ of private keys and a list $VK,VK_{1},\ldots,VK_{l}$ of
verification keys.
\item An encryption algorithm takes as input the public key $PK$, a random string $\omega$
and a clear text $M$. It outputs a ciphertext $c$.
\item A share decryption algorithm takes as input the public key $PK$, an index $1\leq i \leq l$,
the private key $SK_{i}$ and a ciphertext $c$. It outputs a decryption share $c_{i}$ and a proof
of its validity $proof_{i}$.
\item A combining algorithm takes as input the public key $PK$, a ciphertext $c$, a list $c_{1},\ldots,c_{l}$
of decryption shares, the list $VK,VK_{1},\ldots,VK_{l}$ of
verification keys and a list $proof_{1},\ldots,proof_{l}$ of
validity proofs. It outputs a cleartext $M$ or fails.
\end{itemize}

\subsection{The Players and the Scenario}
The game includes the following players: a dealer, a combiner, a set
of $l$ servers $P_{i}$, an adversary and users. All are considered
as probabilistic polynomial time Turing machines, playing in the
following scenario:
\begin{itemize}
\item In an initialization phase, the dealer uses the key generation algorithm to
create the public, private and the verification keys. The public key $PK$ and all
the verification keys $VK,VK_{i}$; where $1\leq i \leq l $ are publicized and each server receives its
shares $SK_{i}$ of the secret key $SK$.
\item To encrypt a message, any user can run the encryption algorithm using the
public key $PK$.
\item To decrypt a ciphertext $c$, the combiner first forwards $c$ to the servers. Using
their secret key $SK_{i}$ and their verification keys $VK,VK_{i}$, each server runs the
decryption algorithm and outputs a partial decryption $c_{i}$ with a proof of validity
of the partial decryption $proof_{i}$. Finally, the combiner uses the combining algorithm
to recover the cleartext, if enough partial decryptions are valid.
\end{itemize}

\subsection{Security requirements}
We recall that an active non-adaptive adversary completely controls the behavior of the corrupted
servers and he chooses which servers he wants to corrupt before key generation.
A threshold cryptosystem is said to be $t-robust$ if the combiner is able to correctly decrypt
any ciphertext, even in the presence of an adversary who actively corrupts up to $t$ servers.
Let us consider an attacker who first issues two messages $M_{0}$ and $M_{1}$. We randomly choose one of these messages. We encrypt it and send this ciphertext to the attacker. Finally, he answers which message has been encrypted. We say that the encryption scheme is semantically secure if there exists no such polynomial time attacker able to guess which of the two messages has been encrypted with a non-negligible advantage.
The semantic security for threshold cryptosystem is defined as follow:

Let be an attacker who actively and non-adaptively corrupts $t$ servers learns the public parameters, the secret keys of the corrupted servers, the public verification keys, all the decryption shares and the proof of validity of those shares.
We consider the following game A:
\begin{itemize}
\item[A1:]The attacker chooses to corrupt $t$ servers. He learns all their secret information and he actively controls their behavior.
\item[A2:]The key generation algorithm is run; the public keys are publicized, each server receives its secret keys and the attacker learns the secrets of the corrupted players.
\item[A3:]The attacker chooses a message $M$ and a partial decryption oracle gives her $l$ valid decryption shares of the encryption of $M$, along with proofs of validity. This step is repeated as many times as the attacker wishes.
\item[A4:]The attacker issues two messages $M_{0}$ and $M_{1}$ and send them to an encryption oracle who randomly chooses a bit $b$ and sends back an encryption $c$ of $M_{b}$ to the attacker.
\item[A5:]The attacker repeats step A3, asking for decryption shares of encryptions of chosen messages.
\item[A6:]The attacker outputs a bit $b'$.
\end{itemize}

A threshold encryption scheme is said to be semantically secure against active
 non-adaptive adversaries if for any polynomial time attacker, $b=b'$ with probability only negligibly greater than $1/2$.
\section{The Boneh-Goh-Nissim Cryptosystem}
\label{ecc}
We recall that the BGNC presented in this section is semantically secure under the \emph{subgroub decision.
problem} (SDP) assumption  \cite{boneh}.

\subsection{Description}

\begin{description}
\item[Key Generation:] The public key is $(n,G,g,h)$, where G is a group of order $n$, $n=q_{1}q_{2}$,
$g$ and $u$ are generators of $G$ and $h=u.q_{2}$. The private key is $q_{1}$.
\item[Encryption :] Pick a random $r\in [0,n-1]$. Let $m\in [0,\ldots,T]$ be a message to encrypt.
So the ciphertext is computed as 
\begin{equation}
\label{Eq1}
C=m\times g+r\times h.
\end{equation}
\item[Decryption :]Compute
\begin{equation}
\label{Eq2}
m=log_{q_{1}g}q_{1}C.
\end{equation}
\end{description}
\subsection{Correctness}
The equation \ref{Eq2} can be written in the following form:
\begin{equation}
\label{Eq3}
    mq_1g=q_1C
\end{equation}
We have to prove that the above equation is equivalent to equation \ref{Eq1}. Multiply the two sides of equation \ref{Eq1} by $q_1$ we get:

\begin{align*}
    q_1C&=q_1(m\times g+r\times h)\\
    &=q_1\times m\times g+q_1\times r\times h\\
    &=q_1mg+q_1r(uq_2) \text{ (recall that $h=uq_2$)}\\
    &=mq_1g+ru(q_1q_2)\\
    &=mq_1g+run \text{ (since $n=q_1q_2$)}\\
    &=mq_1g.\text{ (since $n$ is the order of the group $G$ and $u$ is a generator of the group.)}
\end{align*}

\subsection{Distribution Version of BGNC}

\begin{description}
\item[Key Generation:]The dealer chooses an elliptic curve E such that the order
of the group of the point of this curve is $p+1=n \times s$ and such
that $n=q_{1}q_{2}$, $q_{1}$ and $q_{2}$ are two sufficiently large
prime numbers. So there exist a subgroup $G$ of order $n$. Let $g$
and $u$ be two generators of $G$, set $h=q_{2}u$ and $g_{0}=q_{1}g$.
The public key $PK$ will be $PK=(n,G,g,h,g_{0})$ and the secret key
$SK=q_{1}$. It is shared using Shamir's secret sharing scheme: let
$f_{0}=SK$ and, for $i=1,\ldots,t$, $f_{i}$ is randomly chosen in
$\mathbb{Z}_{n}$. Let $f(X)=\sum_{i=0}^{t}f_{i}X^{i}$; the secret
key $SK_{i}$ is $d_{i}=f(i)\mod n$.

The $d_{i}$'s must be different from
$q_{1}$ and $q_{2}$ so that the decryption shares $c_{i}$ may exist.
We have to ensure that $c_{i} \in [0, T]$ for a suitable chosen $T$.

\item[Encryption:]Compute $C=m\times g+r\times h$ as in the BGNC.
\item[Share Decryption Algorithm:]

Let 
$$ d_{i}=q_{1} \alpha_{i}+ \beta_{i}, \quad\beta_{i} < q_{1}$$ 
$\beta_i$ here is the remaining of the division of $d_i$ by $q_1$. $\alpha_i$ is the quotient. Define  
$$\gamma_{i}=q_{1} \alpha_{i}$$ 
Compute
$$c_{i}=log_{q_{1}g}\gamma_{i}C$$ for $i=1,\ldots,l$. The share of
each server will be the couple $(c_{i},g_{i})$. To convince anyone
that the server $i$ has a valid share $d_{i}$, it uses zero
knowledge proof presented in section \ref{zkp}, resulting in a proof
of validity.

\item[Combining Algorithm:]Let $S$ be a set of valid decryption shares
\newline 
$c_{i}, i=1,\dots,t+1$.
Compute $$m=log_{B}A,$$ where $B=\Sigma_{j\in S} \mu_{0,j}^{S}g_{j}$ and
$A=\Sigma_{j\in S}\mu_{0,j}^{S}(c_{i}g_{0}+\beta_{i}C).$
The coefficients $\mu_{i,j}^{S} $ are the Lagrange coefficients defined by
$$\mu_{i,j}^{S}=D \times \frac{\prod_{j'\in S\backslash {j}}(i-j')}{\prod_{j'\in S\backslash {j}}(j-j')} \in \mathbb{Z}$$
for any $i\in \{ 0,\ldots,l\} $ and any $j\in S$.
\end{description}

We will use this Lagrange interpolation formula:
$$D f(i)= \sum_{j\in S}\mu_{i,j}^{S}f(j)\mod n$$
for any $i\in \{ 0,\ldots,l\} $ and any $j\in S$.

So for $i=0$:
$$D f(0)= \sum_{j\in S}\mu_{0,j}^{S}f(j)\mod n.$$
Thus
\begin{equation}
\label{eq:2}
D d_{0}= \sum_{j\in S}\mu_{0,j}^{S}d_{j}.
\end{equation}

We have that

\begin{align*}
    c_{i}=log_{q_{1}g}\gamma_{i}C
    &\Leftrightarrow
    c_{i}q_1g=\gamma_{i}C\\
    &\Leftrightarrow c_{i}g_{0}=\gamma_{i}C \text{ since $g_0=q_1g$}\\
    &\Leftrightarrow c_{i}g_{0}=(d_{i}- \beta_{i}) C\text{ since $\gamma_i=d_{i}- \beta_{i}$}\\
    & \Leftrightarrow c_{i}g_{0}=d_{i}C- \beta_{i} C.
\end{align*}

Finally
\begin{equation}
\label{eq:1}
c_{i}=log_{q_{1}g}\gamma_{i}C \Leftrightarrow c_{i}g_{0}+\beta_{i} C=d_{i}C
\end{equation}

and
$$m=log_{d_{0}g}d_{0}C \Leftrightarrow md_{0}g=d_{0}C.$$
Multiplying the above formula by $D$ we obtain:
$$mDd_{0}g=d_{0}DC,$$
so by formula \ref{eq:2}
$$m\Sigma_{j\in S}\mu_{0,j}^{S}d_{j}g=\Sigma_{j\in S}\mu_{0,j}^{S}d_{j}C $$
and by formula \ref{eq:1}
$$ \Leftrightarrow m\Sigma_{j\in S}\mu_{0,j}^{S}g_{j}=\Sigma_{j\in S}\mu_{0,j}^{S}(c_{i}g_{0}+\beta_{i}C),$$
or
$$mB=A.$$
So the proposed scheme is correct.
\section{Numerical Example}
We have used Sage software to make a numerical example. 

Let $p = 100000000000000003$ and $q = 100000000000000013$.
With those parameters we obtain an  elliptic curve of order:
$ 96000000000000001536000000000\\00003744.$
We choose for illustration purpose $l = 20$ and $ t = 8$.
We have computed two points of order $n=pq$ :
\footnotesize
$$H = (901696848398424513792042624183449898:919534227604146025519668445930592752:1)$$
\normalsize
and
\footnotesize
$$D = (515582362393135083339332053520498553:934365455532976888281795609449803454:1)$$
\normalsize
A generator of the curve is :
$ G = (744114534472669657966133843400531746:77435475791434420844416408483805765:1)$
We chose to encrypt a message $m = 10$.
We obtain the shares:
[320, 15530, 213640, 1457210, 6593360, 22845370, 65658680, 164354090].
We apply the formula stated in the algorithm, we recover the clear message $m = 10$.

\subsection{Details on Sage Commands Used for the Above Example}
To obtain the two prime numbers and the suitable order for the field
we use the following code :
\begin{verbatim}
p = next_prime(100000000000000000); # this will generate
 a prime greater then 100000000000000000.
q = next_prime(p);# a prime number greater than p.
n = p * q;
l = 1;
pp = l*n-1;
while (not(is_prime(pp)) or mod(pp,3)==1):
    l = l+1;
    pp = l*n-1;
\end{verbatim}

\newpage
After running this code, we obtain a prime number $pp$ with $pp+1 =
l*n$ and this will ensure that the following command:
\begin{verbatim}
E = EllipticCurve(GF(pp),[0,1]);
\end{verbatim}
will generate an elliptic curve of order $pp$ as required in our key generation step.
When looking for a generator of the curve, we use
\begin{verbatim}
E.gens();
\end{verbatim}
To obtain a point of order $n$ we generate a random point and test if his order is as desired:
\begin{verbatim}
o = 1;
while (o <> n):
    H = E.random_point();
    o = H.order();
\end{verbatim}
The discrete logarithm in sage when the range of the
logarithm is known is achived in Sage with the command:
\begin{verbatim}
discrete_log_lambda(P,Q,[0,T],operation='+');
\end{verbatim}
Here $Q$ is the base, $T$ is the range and the operation is relative to
the group on which we are working on.

\begin{monTh}
\label{th}
Under the SDP and the ECDLP assumptions, the distributed version of BGNC (DBGNC) is semantically secure against active non-adaptive adversaries.
\end{monTh}

\begin{proof}
We prove using the reduction process. We show that if an adversary can break the semantic
security of the DBGNC, then we can construct an attacker who can break the semantic
security of BGNC. Following this procedure, we have to simulate data received by the
adversary in steps A2, A3 and A5 of the game A. So the proof consists of proving that the data
simulated during the steps of the game A are indistinguishable from real one.

Let us assume the existence of an adversary $\mathcal{ A}$ able to break the semantic security of the threshold
scheme. We now describe an attacker which uses $\mathcal{ A}$ in order to break the semantic security of BGNC. In a first phase the attacker obtains the public key $(n,G,g,h)$ and
he chooses two messages $M_{0}$ and $M_{1}$ which are sent to an encryption oracle who randomly chooses
a bit $b$ and return an encryption $c$ of $M_{b}$. In a second phase the attacker
tries to guess which message has been encrypted.

We now describe how to feed an adversary $\mathcal{ A}$ of the
threshold scheme in order to make a semantic attacker. In step $A1$
of game A, the adversary chooses to corrupt $t$ servers
$P_{1},\ldots,P_{t}$. In the find phase, the attacker first obtains
the public key $PK=(n,G,g,h)$ of the BGNC. He randomly
chooses $t$ values $d_{1},\ldots,d_{t}$ in the range
$\{0,\ldots,n\}$. $T$ is suitably chosen to efficiently compute the
discrete logarithm on the elliptic curve. As $q_{2}$ is unknown
and $T<q_{2}$. We have to simulate the value of $T$. However, as the
computation
efficiency of the function $log$ is well known, 
 a suitable choice of $T$ may be obtained.


In step 2 of game A the attacker sends $(n,G,g,h,d_{1},\ldots,d_{t})$ to $\mathcal{A} $.

During step A3, $\mathcal{ A}$ chooses a message $M$ and send it to
the attacker. He compute a valid encryption of $M$ given by 
$c=r\times h+M\times g$, where $r$ is a
random number. The decryption shares of
the corrupted players are correctly computed using the
$d_{i}$'s:$c_{i}=log_{d_{1}g} \gamma _{i}c$, and $g_{i}=d_{i}g $ for
$i=1,\ldots,t$. The other shares are computed from the following
formula
$$c_{i}=log_{a}b,$$
where $a=Dg_{0}$ and $b=[\sum_{j\in S}\mu_{ij}( \beta_{j}c+g_{0}c_{j})]-D \beta_{i}c $.
We have that $c_{i}=log_{g_{0}} \gamma _{i}c$. So $c_{i}g_{0}= \gamma _{i}c$, and by
multiplying by $D$: $Dc_{i}g_{0}=D \gamma_{i} c=D(d_{i}- \beta_{i})c=Dd_{i}c-D \beta_{i}c$. We have also
$Dd_{i}=\sum_{j\in S}\mu_{ij}d_{j}$,
this is why we chosen the values of $c_{i}$ as in the above formula.

Finally, the adversary returns \[(c,c_{1},\ldots,c_{l},proof_{1},\ldots,proof_{l}).\]

In step A4, $\mathcal{ A}$ chooses and outputs two messages $M_{0}$ and $M_{1}$. The attacker outputs those
two messages as the results of the find phase.

The encryption oracle for the non-threshold ECC scheme chooses a random bit and sends an encryption
$c$  of $M_{b}$ to the attacker. He forwards $c$ to the adversary $\mathcal{ A}$.

Step A5 is similar to step A3. Finally, in step A6, $\mathcal{ A}$ answers a bit $b'$ which is returned by the
attacker in the guess phase.

We have that $d_{1},\ldots,d_{t}$ and the verification keys are
randomly chosen on  the interval $[0,T]$. So the distribution
received by $\mathcal{ A}$ during the key generation step is
indistinguishable from a real one.

Also, the value received in step A3 and A5 are computed from the secret keys and the values
of the non-corrupted servers are randomly chosen on the interval $[0,n]$.

Finally, all the data simulated by the attacker cannot be distinguished from real ones by $\mathcal{ A}$.
Consequently, if there exists a polynomial time adversary $\mathcal{ A}$ able to break the semantic security of the
threshold scheme, we have made an attacker able to break the semantic security of the original ECC scheme.
\end{proof}

\section{A Second Distribution of BGNC}
\label{eccb}
\subsection{The ElGamal Distributed Version}
First we recall the distributed version of the ElGamal elliptic curve
cryptosystem (DEGECC) taken from \cite{Ertaul}. 
Let the parameters of
the cryptosystem be a finite field $GF(p) $, where $p$ is a prime
number, an elliptic curve $E_{p}(a,b)$ and a point $G$ of order $q$
on the group of point of $E$. Then we run the cryptosystem as follow
:

Bob's private key is $n_{B}$ with $0<n_{B}<q$ and the public key is $K_{B}=n_{B}G$.
\begin{enumerate}
\item First we choose a prime number $p>max(M,n)$, and define $a_{0}=M$ to be the message. Then we randomly select
$k-1$ independent coefficients $a_{1},\ldots,a_{k-1}$, with  $0\leq
a_{j}\leq p-1$; which  define the random polynomial $f(x)$ over a Galois prime field $GF(p)$.
\item We compute $n$ shares, $M_{i}=f(x_{i}) \mod p, 1\leq i\leq n$, where $x_{i}$ can be just the public
index $i$ for simplicity. Convert each index $i$ to a point $P_{i}$ on the elliptic curve $E$.
\item Alice picks a random number $r$ and send $rG$ and $P_{i}+rK_{B}$ to Bob with index $t$.
\item Bob recovers each elliptic curve point by calculating $P_{i}+rK_{B}-n_{B}rG=P_{i}$.
\item Bob converts $P_{i}$ to $M_{i}$ and deduces $M$ by using Lagrange interpolation formula.
\end{enumerate}

Our novel distributed crytosystem runs as follows: The public and
private data are as in section \ref{ecc}.
\begin{enumerate}
\item First Define $a_{0}=M$, the message. Then we select
$k-1$ random, independent coefficients $a_{1},\ldots,a_{k-1},0\leq
a_{j}\leq p-1$, defining the random polynomial $f(x)$ over $GF(p) $.
\item We compute $n$ shares, $M_{i}=f(x_{i}) \mod p, 1\leq i\leq n$, where $x_{i}$ can be just the public
index $i$ for simplicity.
\item Alice picks a random number $r\in [0,n-1]$. So the ciphertext is computed as
$$C_{i}=M_{i}\times g+r\times h.$$
\item Bob recovers each elliptic curve point by calculating $$M_{i}=log_{q_{1}g}q_{1}C_{i}.$$
\item Bob deduces $M$ by using Lagrange interpolation formula.
\end{enumerate}

If an adversary wishes to cryptanalyse our scheme, he will have to
break the SDP assumption, which has not been broken to the best of
our knowledge.

\subsection{Implementation Details and Computation Efficiency}
In this section we compare our second version of the distributed BGNC to DEGECC. We focus on
the ciphering plus the deciphering time of the two cryptosytems. We have
used Sage software \cite{sage} to compare the computational
efficiency of the two algorithms. Rather then implementing the hole
algorithm, we restricted to the encryption plus the decryption
functions, as the other parts of the two algorithms are the same. As
we deal with the computation time, we take the average value of all
the random values $r$. Also, we have chosen $q/2$ as the private key
 $n_{B}$.
 
For the BGNC, we have used $T=100$. So, for a given share
$M_{i}<p$, write $M_{i}$ in base 100. Thus, we obtain
$(M'_{1},\ldots,M'_{s})$, $s$ depending on the value of $p$. Hence,
for all $0<i<s+1$, $0<M'_{i}<100$.

W compute the number of point addition during the ciphering plus the deciphering process
of the two cryptosystems. This is because this operation is the more expensive in time computation.
So, for the DEGECC, we will need $r+r+1+n_{B}+1=2r+2+n_{B}$ point addition for encryption and
decryption of one single point $P_{i}$. For the BGNC, we will need $r+M_{i}$ point addition
plus $s/2$ computation of the logarithm with the Pollard Lambda algorithm. As the value
of $s$ will have logarithm behavior, DEGECC will be more costly then BGNC for large keys.

We have taken 3 curves of different key lengths from \cite{site} and computed the average time
consumption of 100 runs of the two algorithms. 
We plotted  that and obtained figure \ref{ph} which slows the good performance of our algorithm
over the ElGamal one.

\begin{figure}
\centering
\includegraphics[scale=0.6]{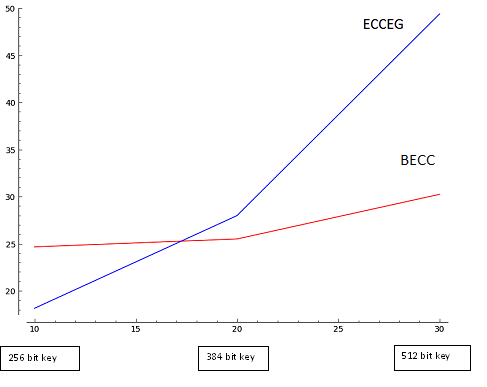}
\caption{Computation time with T}
\label{ph}
\end{figure}

\section{Conclusion}
We have designed two distributions of the BGNC. The first one is proved semantically secure and the second one is more efficient than the DEGECC. While efficiency is considered a poor comparison parameter, it can be an essential one in the setting of a sensor network, wherein energy consumption is of concern, rathen than the security level of the scheme. 











\section{References}

\begin{thebibliography}{99}


\bibitem{Bena}O. Benamara,
{\sl Elliptic Curve Discrete Logarithm Problem},
General Mathematics Notes, vol. 15, pp. 84--91, March 2013.

\bibitem{boneh}D. Boneh, E.J. Goh and N. Kobbi,
{\sl Evaluating 2-DNF formulas on ciphertexts},
Proceedings of the Second international conference on Theory of Cryptography, no. 17, pp. 325--341, 2005.

\bibitem{bro}R. Broker, P. Stevenhagen,
{\sl Constructing of elliptic curves of prime order},
Contemporary Mathematics, vol. 20, 2007.

\bibitem{thresholdecc}P. Changgen, and L. Xiang,
{\sl Threshold Signcryption Scheme Based on Elliptic Curve
Cryptosystem and Verifiable Secret Sharing}, Proc. Intern.
Conference on Wireless Communications, pp. 1182--1185, 2005.

\bibitem{farah}S. Farah, M.Y. Javed, A. Shamim and T. Nawaz,
{\sl An experimental study on Performance Evaluation of Asymmetric
Encryption Algorithms}, Recent Ad. in Information Sciencce, vol. 62,
no. 1, pp. 31--42, 2012.

\bibitem{Farras} O. Farras et al.,
{\sl Linear threshold multisecret sharing schemes}, Inform.
Processing Letters, vol. 112, no. 17-18, pp. 667--673, 2012.


\bibitem{Fouque}P.A. Fouque, G. Poupard and J. Stern,
{\sl Sharing decryption in the context of voting or lotteries},
Financial Cryptography, LNCS,  pp. 90--104, 2000.

\bibitem{Fournaris} A.P. Fournaris,
{\sl A Distributed Approach of a Threshold Certificate-Based Encryption Scheme with No Trusted Entities},
Information Security Journal: A Global Perspective, vol. 22, no. 3, pp. 126--139, 2013.

\bibitem{site}H. Ivey-Law and R. Rolland,
{\sl Constructing a database of cryptographically strong elliptic curves},
Crypto'Puces, 9th-12th May 2011,http://galg.acrypta.com.


\bibitem{Lauter}K. Lauter,
{\sl The advantages of elliptic curve cryptography for wireless security},
Wireless Communications, IEEE, vol. 11, no. 1, pp. 62--67,  2004.

\bibitem{Ertaul}E. Levent and L. Weimin,
{\sl ECC based threshold cryptography for secure data forwarding and secure key exchange in MANET (i)},
Proceedings of the 4th IFIP-TC6 international conference on Networking Technologies, Services, and Protocols; Performance of Computer and Communication Networks; Mobile and Wireless Communication Systems, NETWORKING 05, no.12, pp. 102--113, 2005.

\bibitem{Paillier}P. Paillier,
{\sl Public-key cryptosystems based on composite degree residuosity classes},
Proceedings of the 17th international conference on Theory and application of cryptographic techniques, Springer-Verlag, Berlin, Heidelberg, no. 16, pp. {223--238, 1999.

\bibitem{sage}A. Stein and others,
{\sl Sage Mathematics Software (Version 5.11)},
The Sage Development Team,  2013.

\bibitem{zhang}X. Zhang, F. Zhang, Z. Qin, J. Liu,
{\sl ECC Based threshold decryption scheme and its application in web security},
Journal of Electronic Science and Technology of China, vol. 2, no. 4,  2004.
}

\end{thebibliography}

\end{document}